\documentclass{scrartcl}

\newcommand{\tags}{\usetag{arxiv,full}}

\usepackage{tagging}
\tags

\tagged{soda}{
	\usepackage[T1]{fontenc}
}
\usepackage[margin=1in]{geometry}

\usepackage{hyperref}

\usepackage{algorithm}
\usepackage{algpseudocode}
\tagged{twocol}{\algrenewcommand\algorithmicindent{1.0em}}

\usepackage{amsmath}
\usepackage{amssymb}
\tagged{arxiv}{
	\usepackage{amsthm}
}
\usepackage[capitalise]{cleveref}

\usepackage{csquotes}
\usepackage{enumitem}
\usepackage{footnote}
\usepackage{booktabs}

\tagged{arxiv}{
	\newtheorem{theorem}{Theorem}[section]%
	\newtheorem{lemma}[theorem]{Lemma}
	\newtheorem{corollary}[theorem]{Corollary}
	\newtheorem{Definition}[theorem]{Definition}
}

\newtheorem{claim}{Claim}[section]
\tagged{soda}{
	\crefname{@theorem}{Theorem}{Theorems}
	\crefname{Definition}{Definition}{Definitions}
}

\newlist{inenum}{enumerate*}{1}
\setlist[inenum]{label=(\roman*)}

\makesavenoteenv{tabular}
\makesavenoteenv{table}

\newcommand{\ie}{i.~e.}

\newcommand{\propcomp}[1]{%
	\overline{#1}
}
\newcommand{\propfar}[2]{%
	\propcomp{#1}_{#2}
}
\newcommand{\sizepropcomp}[2]{%
	\propcomp{#1}_{#2}
}
\newcommand{\sizepropfar}[3]{%
	\propfar{#1}{#2;#3}
}
\newcommand{\picomp}{%
	\propcomp{\Pi}
}
\newcommand{\npi}{%
	\Pi_n
}
\newcommand{\npicomp}{%
	\sizepropcomp{\Pi}{n}
}
\newcommand{\piepsfar}{%
	\propfar{\Pi}{>\varepsilon}
}
\newcommand{\npiepsfar}{%
	\sizepropfar{\Pi}{n}{>\varepsilon}
}

\newcommand{\queriessymb}{%
	q
}
\newcommand{\queries}[1]{%
	\queriessymb_{#1}
}
\newcommand{\pqueries}{%
	\queries{\Pi}
}
\newcommand{\cqueriessymb}{%
	\hat{\queriessymb}
}
\newcommand{\cqueries}[1]{%
	\cqueriessymb_{#1}
}
\newcommand{\pcqueries}{%
	\cqueries{\Pi}
}

\newcommand{\defeq}{%
	:=%
}

\newcommand{\N}{\setn}
\newcommand{\setn}{%
	\mathbb{N}%
}
\newcommand{\R}{\setr}
\newcommand{\setr}{%
	\mathbb{R}%
}

\newcommand{\norm}[1]{%
	\lVert #1 \rVert_1%
}
\newcommand{\bignorm}[1]{%
	\big\lVert #1 \big\rVert_1%
}
\newcommand{\biggnorm}[1]{%
	\bigg\lVert #1 \bigg\rVert_1%
}

\newcommand{\disk}[1]{%
	\mbox{$#1$-disc}%
}
\newcommand{\disks}[1]{%
	\mbox{$#1$-discs}%
}
\newcommand{\kdisk}{%
	\disk{k}%
}
\newcommand{\kdisks}{%
	\disks{k}%
}

\newcommand{\fdisk}[3]{%
	\mathrm{disc}_{#1}(#2, #3)%
}
\newcommand{\fkdisk}[2]{%
	\fdisk{k}{#1}{#2}%
}

\newcommand{\diskfreq}[2]{%
	\mathrm{freq}_{#1} \! \left( #2 \right)%
}
\newcommand{\kdiskfreq}[1]{%
	\diskfreq{k}{#1}%
}

\newcommand{\kdiskfreqsub}[2]{%
	\kdiskfreq{#2 \mid #1}%
}

\newcommand{\kdiskfreqent}[2]{%
	\kdiskfreq{#1}_{#2}%
}

\newcommand{\kdiskfreqsubent}[3]{%
	\kdiskfreq{#2 \mid #1}_{#3}%
}

\newcommand{\diskset}[1]{%
	\mathcal{T}_{#1}%
}
\newcommand{\kdiskset}{%
	\diskset{k}%
}

\newcommand{\numdiskssymb}{%
	N%
}
\newcommand{\numdisks}[2]{%
	\numdiskssymb(#1, #2)%
}

\newcommand{\numdkdisks}{%
	\numdisks{d}{k}%
}

\newcommand{\alonsize}[3]{%
	M_{#1}(#2,#3)%
}
\newcommand{\dalonsize}[2]{%
	\alonsize{d}{#1}{#2}
}

\newcommand{\ddkalonsize}{%
	\alonsize{d}{\delta}{k}%
}

\newcommand{\dfp}[2]{%
	$(#1, #2)$-DFP%
}
\newcommand{\dkdfp}{%
	\dfp{\delta}{k}%
}
\newcommand{\bug}[2]{%
	$(#1, #2)$-blow-up graph%
}

\title{Every Testable (Infinite) Property of Bounded-Degree Graphs Contains an Infinite Hyperfinite Subproperty
\thanks{First and third author acknowledge the support by ERC grant No. 307696.}}
\author{
	Hendrik Fichtenberger
	\thanks{Department of Computer Science, TU Dortmund. Email: 	hendrik.fichtenberger@tu-dortmund.de. \href{https://orcid.org/0000-0003-3246-5323}{ORCID~iD: 0000-0003-3246-5323}}
	\and
	Pan Peng
	\thanks{Department of Computer Science, University of Sheffield. Email: p.peng@sheffield.ac.uk.}
	\and
	Christian Sohler
	\thanks{Department of Computer Science, TU Dortmund. Email: christian.sohler@tu-dortmund.de.}
}
\date{}

\begin{document}
	\maketitle
	
	\tagged{soda}{

		\fancyfoot[R]{\scriptsize{Copyright \textcopyright\ 2019\\
		Copyright for this paper is retained by authors}}

	}

	\begin{abstract}
		\tagged{soda}{\small\baselineskip=9pt}
		One of the most fundamental questions in graph property testing is to characterize the combinatorial structure of properties that are testable with a constant number of queries. We work towards an answer to this question for the bounded-degree graph model introduced in \cite{GR02:testing}, where the input graphs
		have maximum degree bounded by a constant $d$. In this model, it is known (among other results) that every \emph{hyperfinite} property is constant-query testable \cite{NS13:hyperfinite}, where, informally, a graph property is hyperfinite, if for every $\delta >0$ every graph in the property can be partitioned into small connected components by removing $\delta n$ edges.
		
		In this paper we show that hyperfiniteness plays a role in \emph{every} testable property, i.e. we show that every testable property is either 
		finite (which trivially implies hyperfiniteness and testability) or contains an infinite hyperfinite subproperty. A simple consequence of our result is that no
		infinite graph property that only consists of expander graphs is constant-query testable.
		
		Based on the above findings, one could ask if every infinite testable non-hyperfinite property might contain an infinite family of expander (or near-expander) graphs. We show that this is not true. Motivated by our counter-example we develop a theorem that shows that we can partition the set of vertices of every bounded
		degree graph into a constant number of subsets and a separator set, such that the separator set is small and the distribution of \kdisks\ on every subset of
		a partition class, is roughly the same as that of the partition class if the subset has small expansion.
	\end{abstract}

\section{Introduction}

Understanding the structure of very large graphs like social networks or the webgraph is a challenging task. Given the size of these networks, it is often
hopeless to compute structural information exactly. A feasible approach is to design random sampling algorithms that only inspect a small portion of the
graph and derive conclusions about the structure of the whole graph from this random sample. However, there are different ways to sample from graphs 
(random induced subgraphs, random sets of edges, random walks, random BFS, %
etc.) and also many structural graph properties. This raises the question,
which sampling approaches (if any) are suitable to detect or approximate which structural properties.

Graph property testing provides a formal algorithmic framework that allows us to study the above setting from a complexity theory point of view. 
In this framework, given oracle access to an input graph, our goal is to distinguish between the case that the graph satisfies 
some property or  that it is ``far from'' having the property by randomly sampling from the graph. Here, a graph property denotes a set of graphs that is invariant under graph isomorphism.
Both oracle access and the notion ``far from'' depend on the representation of the graph. Several models have been proposed in the past two 
decades for dealing with different types of graphs~(see the recent book \cite{Gol17}). 

For dense graphs, Goldreich et al.~\cite{GGR98:testing} introduced the \emph{adjacency matrix} model, in which the algorithm can perform any \emph{vertex-pair} query to the oracle. That is, upon an input vertex pair $u,v$, the oracle returns $1$ if there is an edge between $u,v$ and $0$ otherwise. A graph is called $\varepsilon$-far from having a property $\Pi$ if one has to modify more that $\varepsilon n^2$ vertices to make it satisfy $\Pi$ for any small constant $\varepsilon$. Since the time when the model was introduced, many properties $\Pi$ were found to be \emph{testable} in the sense that there exists an algorithm, called \emph{tester}, that can distinguish if a graph satisfies $\Pi$ or is $\varepsilon$-far from having $\Pi$ while only making a \emph{constant} number of queries. The research in this model has culminated in the seminal work by Alon et al.~\cite{AFNS09:characterization}, who gave a full characterization of constant-query testable properties by the regularity lemma. 

Our understanding of property testing for sparse graphs (e.g., bounded degree graphs) is much more limited. Goldreich and Ron~\cite{GR02:testing} initiated the study of property testing for bounded degree graphs in the \emph{adjacency list} model. A graph $G$ is called a \emph{$d$-bounded} graph if its maximum degree is at most $d$, which is assumed to be a constant. The property tester for a $d$-bounded graph is given oracle access to the adjacency list of the graph, that is, upon an input $(u,i)$ such that $i\leq d$, the oracle returns the $i$-th neighbor of $u$ if such a neighbor exists, and a special symbol otherwise. A $d$-bounded graph is said to be $\varepsilon$-far from having the property $\Pi$ if one needs to modify more than $\varepsilon dn$ edges to obtain a graph that satisfies $\Pi$. In this model, there exist several properties that are known to be testable with a constant number of queries (see discussion below). There also exist a number of properties that require $\tilde{O}(\sqrt{n})$ or $\tilde{O}(n^{1/2+c})$ queries, including bipartiteness~\cite{GR99:bipartiteness}, expansion~\cite{GR00:expansion,CS10:expansion,NS10:expansion,KS11:expansion}, $k$-clusterability~\cite{CPS15:cluster} and one-sided error minor-freeness~\cite{CzuFin14,FicSub18,KSS18:minor}. For the property of being $3$-colorable there is a known $\Omega(n)$ lower bound on the number of queries needed to test the property~\cite{BOT02:color}.

One of the most important questions in this area is to give a purely combinatorial characterization of which graph properties are testable with a constant number of queries. Goldreich and Ron were the first to show that a number of fundamental graph properties including connectivity, $k$-edge connectivity, subgraph-freeness, cycle-freeness, Eulerian and degree regularity can be tested with constant queries in bounded degree graphs~\cite{GR02:testing}. A number of properties with small separators are now known to be testable in a constant number of queries, such as minor closed properties \cite{BSS10:minor,HKNO09:local}, and hyperfinite properties~\cite{NS13:hyperfinite}. In particular, in the latter work it is proved that every property is constant-query testable in hyperfinite graphs. There are also constant-query properties that are closed under edge insertions, including $k$-vertex connectivity~\cite{YI12:vertex}, perfect matching~\cite{YYI12:constant}, sparsity matroid~\cite{ITY12:matroid} and the supermodular-cut condition~\cite{TY15:supermodular}. Furthermore, there exist global monotone properties\footnote{A graph property is called monotone if it is closed under edge deletions.} that contain expander graphs and can be tested with constant queries, including the property of being subdivision-free~\cite{KY13:subdivision}. There also exist some work on testable properties in some special classes of bounded degree graphs. For example, it is known that every hereditary property\footnote{A graph property is called hereditary if it is closed under vertex deletions.} is testable with a constant number of queries in non-expanding $d$-bounded graphs~\cite{CSS09:hereditary}. A property called $\delta$-robust spectral property is constant-query testable in the class of high-girth graphs~\cite{CKSV17:spectra}. However, very little is known about characteristics of all testable properties in general.

\subsection{Our Results}

Although many properties are known to be constant-query testable in bounded degree graphs, our knowledge on characteristics of \emph{all} testable properties is fairly restricted. One prominent example of testable properties is the family of hyperfinite properties \cite{NS13:hyperfinite}, which includes planar graphs and graphs that exclude any fixed minor (see e.g.,~\cite{BSS10:minor,HKNO09:local}). For the statement of our results and the discussion of techniques, we state the definition of hyperfinite graphs at this place.

\begin{Definition}\label{def:hyperfinite}
	Let $\varepsilon\in (0,1]$ and $k\geq 1$. A graph $G$ with maximum degree bounded by $d$ is called \emph{$(\varepsilon,k)$-hyperfinite} if one can remove at most $\varepsilon d |V(G)|$ edges from $G$ so that each connected component of the resulting graph has at most $k$ vertices. For a function $\rho: \R^+\rightarrow\N^+$, a graph $G$ is called \emph{$\rho$-hyperfinite} if $G$ is $(\varepsilon,\rho(\varepsilon))$-hyperfinite for every $\varepsilon>0$. A set (or property) $\Pi$ of graphs is called \emph{$\rho$-hyperfinite} if every graph in $\Pi$ is $\rho$-hyperfinite. A set (or property) $\Pi$ of graphs is called \emph{hyperfinite} if it is $\rho$-hyperfinite for some function $\rho$.
\end{Definition}

Also, many testable properties are known that are not hyperfinite. Our main result is that, nevertheless, for infinite properties the existence of an infinite set of hyperfinite graphs in the property is a necessary condition for its constant-query testability (finite properties are trivially hyperfinite). Since some of these testable properties, e.g., subdivision-freeness, contain expander graphs, a hyperfinite subproperty might seem somewhat surprising. (A subproperty of a property $\Pi$ is a subset of graphs in $\Pi$ that is also invariant under graph isomorphism.) Indeed, the complement of every non-trivially constant-query testable property also contains hyperfinite graphs, where a property is non-trivially testable if it is testable and there exists an $\varepsilon>0$ such that there is
an infinite number of graphs that are $\varepsilon$-far from $\Pi$.

\begin{theorem}
	\label{thm:main_subproperty}
	Every constant-query testable property $\Pi$ of bounded-degree graphs is either finite or contains an infinite hyperfinite subproperty. 
	Also, the complement of every non-trivially constant-query testable graph property contains an infinite hyperfinite subproperty.
\end{theorem}

To our best knowledge, our theorem gives the first non-trivial result on the combinatorial structure of \emph{every} constant-query testable property in bounded-degree graphs. A direct corollary from our main result is that expansion and the $k$-clusterability property are not constant-query testable, as any hyperfinite graph will have many small subsets with small expansion and thus does not satisfy the properties. Indeed, a much stronger lower bound of $\Omega(\sqrt{n})$ on the query complexity for testing these two properties was already known prior to this work~\cite{GR00:expansion}. However, our result further implies that every infinite intersection of a family of expander graphs with any other property is also not testable.

\begin{corollary}
	Let $\Pi$ be a property that does not contain an infinite hyperfinite subproperty, and let $\Pi'$ be an arbitrary property such that $\Pi \cap \Pi'$ is an infinite set. Then, $\Pi \cap \Pi'$ is not testable.
\end{corollary}

Note that in general, the intersection of a property that is not constant-query testable with another property may be testable. For example, the property of being planar and bipartite is testable since it is a hyperfinite property \cite{NS13:hyperfinite}. However, bipartiteness is not constant-query testable \cite{GR02:testing}.

We then study the question whether a similar result can be obtained for \emph{expander} or near-expander subproperties in testable non-hyperfinite properties. Expander graphs are those that are well connected everywhere, and thus can be thought as anti-hyperfinite graphs. Indeed, many known testable, while non-hyperfinite, properties do contain infinite expander subproperties. Typical examples include $k$-connectivity, subgraph-freeness and subdivision-freeness. However,
this turns out to not be the case in general. We show that there exists a testable property that is not hyperfinite and every graph in the property 
has distance $\Omega(n)$ to being an expander graph: The property consists of all graphs that have a connected component on $\lceil |V|/2\rceil$ vertices
and all other vertices are isolated. 

\begin{theorem}
\label{theorem:no-expander}
There exists an infinite graph property $\Pi$ of bounded-degree graphs such that
\begin{itemize}
\item
$\Pi$ is testable (with $2$-sided error) with query complexity $O(d/\varepsilon^2)$,
\item
$\Pi$ is not hyperfinite,
\item
every graph in $\Pi$ differs in $\Omega(n)$ edges from every connected graph.
\end{itemize}
\end{theorem}

Motivated by the above result we also obtain a theorem (\cref{thm:partitionoing}) that shows that we can partition the set of vertices of every bounded
degree graph into a constant number of subsets and a separator set, such that the separator set is small and the distribution of \kdisks\ on every subset of
a partition class, is roughly the same as that of the partition class if the subset has small expansion.

\subsection{Our Techniques}

It is well known that constant-time property testing in the bounded-degree graph model is closely connected to the distribution of \kdisk\
isomorphism types (see, for example, \cite{BSS10:minor,NS13:hyperfinite}). 
The \kdisk\ of $v \in V$ is the rooted subgraph that is induced by all vertices at distance at most $k$ from $v$ and has root $v$, i.e. the local
subgraph that can be explored by running a BFS upto depth $k$. Thus, the distribution of \kdisk\ isomorphism types describes the local structure 
of the graph. We then show (in Theorem~\ref{thm:canonical_tester}) that every constant-query property tester can be turned into a canonical tester that is based on approximating the \kdisk\ 
distribution and decides based on a net over the space of all distribution vectors. Technically, our proof for this result mostly follows an earlier construction of
canonical testers introduced in \cite{GR11:proximity} (see also \cite{CPS16:testing,MMPS17}). 

We then exploit a result by Alon \cite[Proposition 19.10]{Lov12:large} that is derived from open questions in graph limits theory.
Alon proved that for every bounded-degree graph $G$, there exists a graph of constant size $H$ whose \kdisk\ distribution can be made arbitrarily 
close (in terms of $\ell_1$ norm distance) to the \kdisk\ distribution of $G$. Given a graph $G$ on $n$ vertices from some constant-query testable property $\Pi$ we can use multiple copies of
$H$ to obtain a graph that consists of connected components of constant size and whose distribution of \kdisks\ is close to that of $G$. 
The latter implies that a canonical tester will behave similarly on $H$ and $G$ and thus accepts with probability at least $2/3$. Although $H$ does not 
necessarily have the tested property, it must be close to it. This implies that there exists a graph $H'$ in $\Pi$ from which we can remove
$\varepsilon d n$ edges to partition it into small connected components. Thus, $H'$ is $(\varepsilon,O_\varepsilon(1))$-hyperfinite, where $O_\varepsilon(1)$ is a constant depending on $\varepsilon$. 
However, $H'$ may not be $(\varepsilon', O_{\varepsilon'}(1))$-hyperfinite for $\varepsilon' < \varepsilon$. The challenge is how to construct such a graph.

In order to do so, we proceed as follows. For every suitable choice of $n$, we construct a series of $n$-vertex graphs $H_i$ such that each $H_i$ approximately inherits the $(\varepsilon, O_\varepsilon(1))$-hyperfinite properties of all graphs $H_{i'}$ for all $i' < i$. The key idea is to maintain the hyperfinite properties of $H_i$ by causing only a small perturbation of its \kdisk\ vector. Carefully choosing the parameters of this process, at the end we obtain a graph $H^{(n)}$ that is $\rho(\varepsilon)$-hyperfinite for a monotone function $\rho(\cdot)$ and every $\varepsilon > 0$.

In order to show that we cannot obtain a similar result for expander graphs in non-hyperfinite properties, we have designed the aforementioned
property of graphs which consist of a connected component on half of the vertices and all other vertices are isolated. Our proof of testability
combines earlier ideas of testing connectivity with simple sampling based estimation of the number of isolated vertices.

\subsection{Other Related Work}%
Goldreich and Ron~\cite{GR11:proximity} gave characterizations of the graph properties that have constant-query proximity-oblivious testers for bounded-degree graphs and for dense graphs. As noted in~\cite{GR11:proximity}, such a class of properties is a rather restricted subset of the class of all constant-query testable properties. Hyperfiniteness is also closely related to \emph{graphings} that have been investigated in the theory of graph limits~\cite{Ele07:note,Sch08:hyperfinite,Lov12:large}. 

\section{Preliminaries}
Let $G=(V,E)$ be a graph with maximum degree bounded by $d$, which is assumed to be a constant. We also call $G$ a $d$-bounded graph.

\begin{Definition}
	A \emph{graph property} $\Pi$ is a set of graphs that is invariant under graph isomorphism. If all the graphs in $\Pi$ have maximum degree upper bounded by $d$, then we call $\Pi$ a $d$-bounded graph property. We let $\npi \subseteq \Pi$ denote the set of graphs in $\Pi$ with $n$ vertices. Note that $\Pi=\cup_{n\geq 1}\npi$. 
	
	Let $\picomp$ denote the complement of $\Pi$, i.e., $\picomp=\mathcal{U}\setminus \Pi$, where $\mathcal{U}$ denotes the set of all $d$-bounded graphs. Let $\npicomp$ denote the set of $n$-vertex graphs that are not in $\npi$, i.e., $\npicomp=\mathcal{U}_n\setminus \Pi_n$, where $\mathcal{U}_n$ denotes the set of all $d$-bounded $n$-vertex graphs. 
	
	A subset $\Pi'\subseteq \Pi$ is called a \emph{subproperty} of $\Pi$ if $\Pi'$ is  invariant under graph isomorphism.
\end{Definition}

We have the following definition on graphs that are far from having some property.
\begin{Definition}
	Let $\Pi=\cup_{n\geq 1}\npi$ be a $d$-bounded graph property. An $n$-vertex graph is said to be $\varepsilon$-far from having property $\npi$ if one has to modify more than $\varepsilon dn$ edges to make it satisfy $\npi$.
	
	Let $\npiepsfar$ denote the set of all $n$-vertex graphs that are $\varepsilon$-far from $\npi$. Let $\piepsfar \subseteq \picomp$ be the set of all graphs that are $\varepsilon$-far from $\Pi$, i.e., $\piepsfar=\cup_{n \geq 1}\npiepsfar$.
\end{Definition}

Given a property $\Pi=\cup_{n \geq 1} \npi$, an algorithm is called a \emph{tester} for $\Pi$, if it takes as input parameters $0<\varepsilon\leq 1, n, d$, and has query access to the adjacency lists of an $n$-vertex $d$-bounded graph $G$, and with probability at least $2/3$, accepts $G$ if $G \in \npi$ and rejects $G$ if $G \in \npiepsfar$. The following gives the definition of constant-query testable properties.

\begin{Definition}
	We call a $d$-bounded graph property $\Pi=\cup_{n\geq 1}\npi$ \emph{(constant-query) testable}, if there exists a tester for $\Pi$ that makes at most $\queries{\Pi} = \queries{\Pi}(\varepsilon, d)$ queries for some function $\queries{\Pi}(\cdot, \cdot)$ that depends only on $\varepsilon, d$.
\end{Definition}

\paragraph{$k$-Discs and frequency vectors.}%
The notions of \kdisks\ and frequency vectors play an important role for analyzing constant-query testable properties. For any vertex $v\in V$, we let $\fkdisk{G}{v}$ denote the subgraph rooted at $v$ that is induced by all vertices that are at distance at most~$k$ from~$v$. For any two rooted subgraphs $H_1,H_2$, we say $H_1$ is isomorphic to $H_2$, denoted by $H_1\simeq H_2$, if there exists a root-preserving mapping $\Phi:V(H_1)\rightarrow V(H_2)$ such that $(u,v)\in E(H_1)$ if and only if $(\Phi(u),\Phi(v))\in E(H_2)$. Note that for constant $d$, the total number of possible non-isomorphic \kdisks\ is also a constant, denoted by $\numdkdisks$. Furthermore, we let $\kdiskset=\{\Delta_1,\cdots,\Delta_\numdiskssymb\}$ be the set of all isomorphism types of \kdisks\ in any $d$-bounded graph, where $\numdiskssymb=\numdkdisks$. Finally, we let $\kdiskfreq{G}$ denote the \emph{frequency vector} of $G$ which is indexed by \kdisk\ types in $\kdiskset$ such that 
$$\kdiskfreqent{G}{\Delta}=\frac{|\{v\in V: \fkdisk{G}{v}\simeq \Delta \}|}{n}$$ 
for any $\Delta\in\kdiskset$, i.e., $\kdiskfreqent{G}{\Delta}$ denotes the fraction of vertices in $G$ whose \kdisks\ are isomorphic to $\Delta$. Furthermore, for any subset $S$ of $G$, we let $\kdiskfreqsub{G}{S}$ denote the vector that is indexed by types in $\kdiskset$ such that 
$$\kdiskfreqsubent{G}{S}{\Delta}=\frac{|\{v\in S: \fkdisk{G}{v}\simeq \Delta \}|}{|S|}$$ 
for any $\Delta\in\kdiskset$, i.e., $\kdiskfreqsubent{G}{S}{\Delta}$ denotes the fraction of vertices in $S$ whose \kdisks\ in $G$ are isomorphic to $\Delta$. Note that $\kdiskfreq{G} = \kdiskfreqsub{G}{V}$. If $S$ contains a single element $x$, we write $\kdiskfreqsub{G}{x} = \kdiskfreqsub{G}{S}$.

For any vector $f$, we let $\norm{f}$ denote its $\ell_1$-norm. We have the following simple lemma on the $\ell_1$-norm distance of the frequency vectors of two graphs that are $\varepsilon$-close to each other. The proof follows from the proof of Corollary 3 in~\cite{FPS15:constant}, while we provide a proof here for the sake of completeness.
\begin{lemma}
	\label{thm:close_freq}
	Let $\varepsilon > 0$ and $k\geq 1$. Let $G_1, G_2$ be $d$-bounded graphs such that $G_1$ is $\varepsilon$-close to $G_2$. Then, $\norm{\kdiskfreq{G} - \kdiskfreq{G_2}} < 6\varepsilon d^{k+1}$. %
\end{lemma}
\begin{proof}
	Let $F:=E(G_1)\triangle E(G_2)$ denote the set of edges that appear only in one of the two graphs $G_1,G_2$. Since $G_1$ is $\varepsilon$-close to $G_2$, it holds that $|F|\leq \varepsilon dn$. Note that for any $e\in F$, the total number of vertices that are within distance at most $k$ to either of its endpoint is at most $2(1+d+d(d-1)+\cdots+d(d-1)^{k-1})\leq 3d^k$. This further implies that the total number of vertices that may have different \kdisk\ types in $G_1$ and $G_2$ is at most $|F|\cdot 3d^k \leq 3\varepsilon d^{k+1} n$. Finally, we note that each vertex with different \kdisk\ types in $G_1$ and $G_2$ contributes at most $\frac{2}{n}$ to the $\ell_1$-norm distance of $\kdiskfreq{G_1}$ and $\kdiskfreq{G_2}$, which implies that 
	$$\norm{\kdiskfreq{G_1} - \kdiskfreq{G_2}} < 3\varepsilon d^{k+1} n\cdot \frac{2}{n}=6\varepsilon d^{k+1}.$$
	This completes the proof of the lemma.
\end{proof}

The converse to the above lemma is not true in general, that is, it is not true that the closeness of the frequency vectors of two graphs implies the closeness of these two graphs. However, Benjamini et al.~\cite{BSS10:minor} showed that the converse somehow still holds for hyperfinite graphs. More precisely, they proved the following result.  

\begin{lemma}[Theorem 2.2 from~\cite{BSS10:minor}]\label{lemma:benj}
	Let $d,s\geq 1$ and $\varepsilon>0$. Let $\Lambda_1$ be the set of $(\varepsilon,s)$-hyperfinite $d$-bounded graphs, and let $\Lambda_2$ be the set of $d$-bounded graphs that are not $(4\varepsilon\log(4d/\varepsilon), s)$-hyperfinite. Then it holds that for any graph $G_1 \in \Lambda_1$ and graph $G_2 \in \Lambda_2$, 
	$$\norm{\kdiskfreq{G_1} - \kdiskfreq{G_2}}>\frac{8\varepsilon}{d}\log(4/3),$$
	where %
	$k = 10sd^{2s+1} / \varepsilon$.
\end{lemma}

\paragraph{Frequency preservers and blow-up graphs.} The following lemma is due to Alon, and it roughly says that for any $n$-vertex $d$-bounded graph, there always exists a ``small'' graph whose size is independent of $n$ that preserves the local structure well, i.e., its \kdisk\ frequencies.%

\begin{lemma}[Proposition 19.10 in \cite{Lov12:large}]
	\label{thm:alon}
	For any $\delta>0$ and $d,k\geq 1$, there exists a function $\ddkalonsize$ such that for every $n$-vertex graph $G$, there exists a graph $H$ of size at most $\ddkalonsize$ such that $\norm{\kdiskfreq{G} - \kdiskfreq{H}} < \delta$.
\end{lemma}
\begin{Definition}[\dkdfp]
	We call the small graph $H$ obtained from \cref{thm:alon} a \emph{$(\delta,k)$-disk frequency preserver} (abbreviated as \emph{\dkdfp}) of $G$.
\end{Definition}
We remark that though we know the existence of the function $\ddkalonsize$ that upper bounds the size of some \dkdfp, there is no known explicit bound on $\ddkalonsize$ for arbitrary $d$-bound graphs (see~\cite{FPS15:constant} for explicit bounds of $\ddkalonsize$ for some special classes of graphs).

We use DFPs as a building block to construct $n$-vertex graphs that have constant-size connected components and approximately preserve the \kdisk\ frequencies of a given $n$-vertex graph $G$. More precisely, we have the following definition.
\begin{Definition}[Blow-Up Graph]
	\label{def:blowup_graph}
	Let $\delta, k > 0$, and let $G$ be a $d$-bounded $n$-vertex graph. Let $H$ be a \dkdfp\ of $G$ graph of size $h \leq \dalonsize{\delta}{k}$. %
	Let $H'$ be the $n$-vertex graph that is composed of $\lfloor n / h \rfloor$ disjoint copies of $H$ and $n - h \cdot \lfloor n / h  \rfloor$ isolated vertices. We call $H'$ the \bug{\delta}{k} of~$G$.
\end{Definition}

The following lemma follows directly from the above definition of blow-up graphs and the fact that the blow-up graph contains at most $h\leq \dalonsize{\delta}{k}$ isolated vertices.
\begin{lemma}
	\label{thm:blowup_locality}
	Let $\delta,d, k > 0$. Let $n\geq n_0(\delta,d,k):=20\dalonsize{\delta}{k}/\delta$. Let $G$ be any $d$-bounded $n$-vertex graph and let $H$ be the \bug{\delta}{k} of $G$. We have $\norm{\kdiskfreq{G} - \kdiskfreq{H}} < 1.1\delta$.
\end{lemma}

\paragraph{Expansion and expander graphs.} Let $G=(V,E)$ be a $d$-bounded graph. Let $S\subset V$ be a subset such that $|S|\leq |V|/2$. The \emph{expansion} or \emph{conductance} of set $S$ is defined to be $\phi_G(S)=\frac{e(S,V\setminus S)}{d|S|}$, where $e(S,V\setminus S)$ denotes the number of crossing edges from $S$ to $V\setminus S$. The expansion of $G$ is defined as $\phi(G):=\min_{S:|S|\leq |V|/2}\phi_G(S)$. We call $G$ a \emph{$\phi$-expander} if $\phi(G)\geq \phi$. We simply call $G$ an expander if $G$ is a $\phi$-expander for some universal constant $\phi$.

\section{Constant-Query Testable Properties and Hyperfinite Properties}

In this section, we give the proof of main theorem, i.e., Theorem~\ref{thm:main_subproperty}. We first give the necessary tools in Section~\ref{sec:basic_tools}, and then give the proof of the first part and second part of Theorem~\ref{thm:main_subproperty} in Section~\ref{sec:testable_property} and \ref{sec:complement}, respectively.  %

\subsection{Basic Tools}\label{sec:basic_tools}
The following is a direct corollary of Lemma~\ref{lemma:benj} by Benjamini et al.~\cite{BSS10:minor}.
\begin{lemma}
	\label{thm:local_encondig_hypf}
	Let $\varepsilon, s > 0$. 
	Let $\Pi$ be a testable graph property. Suppose there exists a graph $G \in \Pi_n$ that is $(\varepsilon, s)$-hyperfinite. %
	Then, every graph $G' \in \Pi_n$ such that $\norm{\kdiskfreq{G} - \kdiskfreq{G'}} < \frac{8 \varepsilon}{d} \log(4/3)$ is $(4 \varepsilon \log \frac{4d}{\varepsilon}, s)$-hyperfinite,
	where $k = 10sd^{2s+1} / \varepsilon$.
\end{lemma}

Our second tool is the following characterization of constant-query testable properties by the so-called \emph{canonical tester}. Such a characterization is similar to the previous ones given in~\cite{GR11:proximity,CPS16:testing} for bounded-degree testable graph properties. The main difference here is that our canonical tester makes decisions based on the frequency vectors, instead of the forbidden subgraphs as considered in the previous work. We have the following theorem, whose proof is deferred to \cref{sec:canonical_tester}.

\begin{theorem}[Canonical Tester]
	\label{thm:canonical_tester}
	Let \( \Pi = (\npi)_{n \in \setn} \) be a graph property that can be tested with query complexity \( \pqueries(\varepsilon,d) \). %
Then there exists $t\defeq c\cdot \pqueries(\varepsilon,d)$ for some constant $c>1$, \(  n_1:=n_1(\varepsilon, d)\) such that for any $\varepsilon>0$, $d$, $n\geq n_1$, there exists a tester $\mathcal{T}_C$ that
	\begin{enumerate}
		\item accepts any $d$-bounded $n$-vertex graph $G$ with probability at least $2/3$, if \( \min_{G' \in \npi} \norm{\diskfreq{t}{G} - \diskfreq{t}{G'}} \leq  \frac{1}{12 t} \),\label{ite:canonical_cond_1}
		\item rejects any $d$-bounded $n$-vertex graph $G$ with probability at least $2/3$, if \( \min_{G' \in \npiepsfar} \norm{\diskfreq{t}{G} - \diskfreq{t}{G'}} \leq  \frac{1}{12 t} \).\label{ite:canonical_cond_2}
	\end{enumerate}
	The canonical tester has query complexity $\pcqueries(\varepsilon,d) \leq t\cdot d^{t+2}$. 
\end{theorem}

\subsection{Infinite Testable Property Contain Infinite Hyperfinite Subproperties}\label{sec:testable_property}

We now prove the first part of \cref{thm:main_subproperty}, \ie, every infinite testable property contains an infinite hyperfinite subproperty.

We start by showing that for any \emph{fixed} $\varepsilon$, and any graph $G$ in a testable property $\Pi$, we can find another graph $G'$ such that $G'$ is $(\varepsilon,s)$-hyperfinite and the frequency vectors of $G$ and $G'$ are close. 
\begin{lemma}
	\label{thm:hypf_subproperty_step}
	Let $\delta, \varepsilon, k > 0$. Let $\varepsilon' = \min \{ \varepsilon, \frac{\delta}{18d^{k+1}} \}$. Let $n\geq n_2(\varepsilon,\delta,d,k)$. Let $\Pi$ be a testable graph property with query complexity \( \pqueries = \pqueries(\varepsilon,d) \) and let $G \in \Pi_n$. Then, there exists $G' \in \Pi_n$ such that 
	\begin{itemize}
		\item $G'$ is $(\varepsilon, \dalonsize{\frac{\delta'}{3}}{k'})$-hyperfinite, and
		\item $\norm{\diskfreq{k}{G}-\diskfreq{k}{G'}} < \delta$,
	\end{itemize}
	where $\delta' = \min \{ \delta, \frac{1}{5 c\cdot\pqueries(\varepsilon',d)} \}$ and $k' = \max \{k, c\cdot\pqueries(\varepsilon',d) \}$ for some constant $c>1$.
\end{lemma}
\begin{proof}
Let $n_2(\varepsilon,\delta,d,k)=\max\{n_0(\frac{\delta'}{3},d,k), \linebreak[1] n_1(\varepsilon',d)\}$, where $n_0,n_1$ are the numbers in the statements of \cref{thm:blowup_locality} and	\cref{thm:canonical_tester}, respectively. Let $t=c\cdot\pqueries(\varepsilon',d)$ for the constant $c>1$ from \cref{thm:canonical_tester}. By definition, it holds that $t\leq k'$.  Let 
	$H$ be the $(\frac{\delta'}{3}, k')$-blow-up graph of $G$. By \cref{thm:blowup_locality} and our assumption that $n\geq n_2$, it holds that $\norm{\diskfreq{k'}{G} - \diskfreq{k'}{H}} \leq \frac{1.1\delta'}{3}$, which implies that 
	\begin{eqnarray}
	\norm{\diskfreq{t}{G} - \diskfreq{t}{H}} \leq \frac{1.1\delta'}{3}\leq \frac{1}{12t}, \label{ineq:diskfreq}
	\end{eqnarray}
	as $t$ satisfies that $\frac{1}{5t}\geq \delta'$ and that $t\leq k'$.

	Let $\mathcal{T}_C$ be the canonical tester for $\Pi$ with parameter $\varepsilon'$ with corresponding query complexity $t=\pcqueries(\varepsilon',d)$. Then by \cref{thm:canonical_tester}, $\mathcal{T}_C$ will accept $H$ with probability at least $2/3$. This implies that $H$ is $\varepsilon'$-close to $\Pi$. Let $G' \in \Pi$ such that $H$ is $\varepsilon'$-close to $G'$. We claim that $G'$ is the graph we are looking for.
	
	First, we show that $G'$ is $(\varepsilon, \dalonsize{\frac{\delta'}{3}}{k'}))$-hyperfinite. Recall that by definition, $H$ is composed of $\lfloor n / h \rfloor$ disjoint copies of a graph of size $h$ and $n - h \cdot \lfloor n / h  \rfloor$ isolated vertices, where $h\leq \dalonsize{\frac{\delta'}{3}}{k'}$. This implies that $H$ is $(0, \dalonsize{\frac{\delta'}{3}}{k'})$-hyperfinite. %
	It follows that $G'$ is $(\varepsilon, \dalonsize{\frac{\delta'}{3}}{k'})$-hyperfinite because we can remove at most $\varepsilon' d n \leq \varepsilon d n$ edges from $G'$ to obtain a graph of which all connected components have size at most $\dalonsize{\frac{\delta'}{3}}{k'}$. %
	
	Second, we prove that $\norm{\kdiskfreq{G} - \kdiskfreq{G'}} \leq \delta$. Note that the bound given by inequality~(\ref{ineq:diskfreq}) implies $$\norm{\kdiskfreq{G} - \kdiskfreq{H}} \leq \frac{1.1\delta'}{3}\leq \frac{1.1\delta}{3},$$ as $k \leq k'$ and $\delta \geq \delta'$. Now since $H$ and $G'$ are $\varepsilon'$-close to each other, by \cref{thm:close_freq}, we have that 
	$$\norm{\kdiskfreq{H} - \kdiskfreq{G'}} < 6\varepsilon' d^{k+1}\leq \frac{\delta}{3},$$
	where the last inequality follows from our setting of parameters.  The claim then follows by applying the triangle inequality.
	This completes the proof of the lemma.
\end{proof}

The above lemma only guarantees that for \emph{every fixed} $\varepsilon > 0$, and graph $G\in\Pi_n$, one can find a graph $G_\varepsilon \in \Pi_n$ that is $(\varepsilon, \dalonsize{\delta'}{k'})$-hyperfinite (for $\delta'$ and $k'$ as in \cref{thm:hypf_subproperty_step}). However, we cannot directly use $G_\varepsilon$ to construct an infinite hyperfinite subproperty. Recall that a set $\Pi$ of graphs is called to be a hyperfinite property if there exists a function $\rho: (0, 1] \rightarrow \setn$ such that $\Pi$ is $(\varepsilon, \rho(\varepsilon))$-hyperfinite for every $\varepsilon > 0$. Now, for any $\varepsilon'<\varepsilon$, we cannot guarantee that after removing $\varepsilon' dn$ edge from $G_\varepsilon$, one can obtain a graph that is the union of connected components of constant size. Furthermore, it is not guaranteed that $G_{\varepsilon_1} \simeq G_{\varepsilon_2}$ if $\varepsilon_1 \neq \varepsilon_2$.

Our idea of overcoming the above difficulty is to start with the above hyperfinite graph $G_0:=G_\varepsilon \in \Pi_n$ for some fixed $\varepsilon>0$, and then iteratively construct a sequence of graphs $G_i\in \Pi_n$ with $i\geq 1$ from $G_{i-1}$. The constructed graph $G_{i+1}$ is guaranteed to inherit hyperfinite properties from $G_i$. The key idea is to maintain the hyperfinite properties of $G_i$ by causing only a small perturbation of its \kdisk\ vector. Choosing the parameters in this process carefully, we can maintain these hyperfinite properties for the whole sequence of graphs. Now we give the details in the following lemma. Note that the first part of Theorem~\ref{thm:main_subproperty} follows from this lemma.%

\begin{lemma}
	\label{thm:hypf_subproperty_main}
	Let $\Pi$ be an infinite $d$-bounded graph property that is testable with query complexity $\pqueries(\varepsilon,d)$. Then, there exists $\Pi' \subseteq \Pi$ such that \begin{itemize}
		\item $\Pi'$ is an infinite subproperty of $\Pi$, and
		\item there exists a monotonically decreasing function $\rho : (0, 1] \rightarrow \setn$ such that $\Pi'$ is $(\varepsilon, \rho(\varepsilon))$-hyperfinite for every $\varepsilon > 0$.
	\end{itemize}
\end{lemma}
\begin{proof}
	Let $X:=\{ \lvert V \rvert : G=(V,E) \in \Pi\}$ be the set of sizes $|V(G)|$ of graphs $G$ in $\Pi$. Since $\Pi$ is an infinite graph property, it holds that $X$ is also an infinite set. We show there exists a monotonically decreasing function $\rho \colon (0,1] \rightarrow \setn$ such that for each $n\in X$, we can find a graph $H^{(n)} \in \npi$ that is $(\varepsilon, \rho(\varepsilon))$-hyperfinite for every $\varepsilon > 0$. %
	This will imply that the set $\Pi' = \{H^{(n)}: n\in X \}$ is an infinite $\rho$-hyperfinite property, which will then prove the lemma.

	Let us now fix an arbitrary $n \in X$ and let $G \in \npi$ be an arbitrary graph in $\npi$. We let \textsc{FindHyper}($G,\delta,\varepsilon, k, \Pi_n$) denote the graph $G'$ that is obtained by applying Lemma~\ref{thm:hypf_subproperty_step} on $G\in \Pi_n$ with parameters $\delta,\varepsilon, k$. Now we construct $H^{(n)}$ as follows. %
	
	Let $\varepsilon_1=\frac{1}{10}$. Let $\delta_1 = 4 \varepsilon_1 / d \log (4/3)$ and let $k_1=1$. If $n<n_3(d):=n_2(\varepsilon_1,\delta_1,d,k_1)=n_2(\frac{1}{10},\frac{2}{5d}\log(4/3),d, 1)$, where $n_2$ is the number given in Lemma~\ref{thm:hypf_subproperty_step}, then we simply let $H^{(n)} = G$, which is a finite graph of size at most $n_3$. In the following, we assume that $n\geq n_3$. 
	
	Let $G_0=G$. We start by applying \cref{thm:hypf_subproperty_step} to $G_0$ with parameters $\delta = \delta_1$, $\varepsilon = \varepsilon_1$ and $k = k_1$ to obtain a graph $G_1$ that is $(\varepsilon_1, s_1)$-hyperfinite, where $s_1:=\dalonsize{\frac{\delta_1'}{3}}{k'_1}$ and $\delta_1' = \min \{ \delta_1, \frac{1}{5 c\cdot\pqueries(\varepsilon_1',d)} \}, k_1' = \max \{k_1, c\cdot\pqueries(\varepsilon_1',d) \}$, $\varepsilon_1' = \min \{ \varepsilon_1, \frac{\delta_1}{18d^{k_1+1}} \}$.

	We now iteratively construct a new $n$-vertex graph $G_{i+1}$ from a graph $G_i$ that is  $(\varepsilon_i, s_i)$-hyperfinite, where $s_i:=\dalonsize{\frac{\delta_i'}{3}}{k'_i}$. Let 
	\begin{align*}
		\delta_{i+1} & \defeq \delta_i / 2,
		&
		\varepsilon_{i+1} & \defeq \varepsilon_i / 2,
		&
		\tagged{twocol}{\displaybreak[3]\\}
		k_{i+1} & \defeq \max\{ k_i, 10 s_i d^{2 s_i + 1} / \varepsilon_i \}.
	\end{align*}
	We apply \cref{thm:hypf_subproperty_step} to $G_i$ with parameters $\varepsilon = \varepsilon_{i+1}$, $\delta = \delta_{i+1}$ and $k = k_{i+1}$ to obtain a graph $G_{i+1}$ that is $(\varepsilon_{i+1}, s_{i+1})$-hyperfinite, where $s_{i+1}:=\dalonsize{\frac{\delta_{i+1}'}{3}}{k_{i+1}'}$, and 
	\begin{align*}
		\delta_{i+1}' &= \min \{ \delta_{i+1}, \frac{1}{5 c\cdot\pqueries(\varepsilon_{i+1}',d)} \},
		& \tagged{twocol}{\displaybreak[3]\\}
		k_{i+1}' &= \max \{k_{i+1}, c\cdot\pqueries(\varepsilon_{i+1}',d) \},
		& \tagged{twocol}{\displaybreak[3]\\}
		\varepsilon_{i+1}' &= \min \{ \varepsilon_{i+1}, \frac{\delta_{i+1}}{18d^{k_{i+1}+1}} \}.
	\end{align*}
	Finally, we stop the  process after the $i'$-th iteration such that $\varepsilon_{i'} d n < 1$. We set $H^{(n)} = G_{i'}$. The pseudo-code of the whole process is given in \cref{alg:hyperfinite} (which invokes \cref{alg:setsize} for setting the parameters as a subroutine).
	
	\algtext*{EndWhile}
	\algtext*{EndIf}
	\algtext*{EndFor}
	
	\begin{algorithm}[H]
		\caption{Construction of $H^{(n)}$}~\label{alg:hyperfinite}
		\begin{algorithmic}[1]
			\Procedure{Construct}{$G,\Pi_n$}	
			\State $G_0\gets G$, $\varepsilon_1\gets \frac{1}{10}$\untagged{twocol}{, }
			\tagged{twocol}{\State}$\delta_1 \gets 4 \varepsilon_1 / d \log (4/3)$, $k_1 \gets 1$		
			\State $G_1\gets$\textsc{FindHyper}($G_0,\delta_1,\varepsilon_1, k_1,\Pi_n$)
			\State $s_1\gets$\textsc{SetSize}($\varepsilon_1,\delta_1,k_1,\Pi_n$) %
			\State $i\gets 1$
			\While{$\varepsilon_i dn\geq 1$}
			\State $\varepsilon_{i+1} \gets \varepsilon_i / 2$, $\delta_{i+1} \gets \delta_i / 2$\untagged{twocol}{,}
			\tagged{twocol}{\State}$k_{i+1} \gets \max\{ k_i, 10 s_i d^{2 s_i + 1} / \varepsilon_i\}$
			\State $G_{i+1}\gets$\textsc{FindHyper}($G_i,\delta_{i+1},\varepsilon_{i+1}, k_{i+1},\Pi_n$)~%
			\State $s_{i+1}\gets$\textsc{SetSize}($\varepsilon_{i+1},\delta_{i+1},k_{i+1},\Pi_n$)
			\State{$i\gets i+1$}
			\EndWhile
			\State{\Return $H^{(n)}\gets G_i$}
			\EndProcedure
		\end{algorithmic}
	\end{algorithm}

\begin{algorithm}[H]
	\caption{Set the value of $s$}~\label{alg:setsize}
	\begin{algorithmic}[1]
		\Procedure{SetSize}{$\varepsilon,\delta,k,\Pi_n$}	
		\State $\varepsilon' \gets \min \{ \varepsilon, \frac{\delta}{18d^{k+1}} \}$\untagged{twocol}{, }
		\tagged{twocol}{\State}$\delta' \gets \min \{ \delta, \frac{1}{5c\cdot \pqueries(\varepsilon',d)} \}$\untagged{twocol}{, }
		\tagged{twocol}{\State}$k' \gets \max \{k, c\cdot\pqueries(\varepsilon',d) \}$
		\State $s\gets \dalonsize{\frac{\delta'}{3}}{k'}$
		\State{\Return $s$}
		\EndProcedure
	\end{algorithmic}
\end{algorithm}

	Now we also note that by the construction and Lemma~\ref{thm:hypf_subproperty_step}, it holds that for any $i\geq 0$, 
	$$\norm{\diskfreq{k_{i+1}}{G_{i+1}} - \diskfreq{k_{i+1}}{G_i}}<\delta_{i+1}.$$
	By noting that $k_j\leq k_{i+1}$ %
	for any $j\leq i+1$, we have that 
	\begin{alignat*}{10}
		& \norm{\diskfreq{k_{j}}{G_{i+1}} - \diskfreq{k_{j}}{G_i}} \tagged{twocol}{\displaybreak[3]\\}
		\leq {}& \norm{\diskfreq{k_{i+1}}{G_{i+1}} - \diskfreq{k_{i+1}}{G_i}} &
		< {}& \delta_{i+1}. &
	\end{alignat*}
	
	Furthermore, we have the following claim.
	\begin{claim}\label{claim:freq_diff}
		It holds that $\norm{\diskfreq{k_j}{G_{i+1}} - \diskfreq{k_j}{G_j}} < 8 \varepsilon_j / d \log(4/3)$ for all $j \leq i+1$.
	\end{claim} 
	\begin{proof}
		Recall that $\varepsilon_{i+1} = \varepsilon_i / 2$ and $\delta_{i+1} = \delta_i / 2$ for all $i > 1$. We have
		\begin{alignat*}{10}
			& \norm{\diskfreq{k_j}{G_{i+1}} - \diskfreq{k_j}{G_j}} & \tagged{twocol}{\displaybreak[3]\\}
			\leq {}& \sum_{\ell=j}^{i} \norm{\diskfreq{k_j}{G_{\ell+1}} - \diskfreq{k_j}{G_\ell}} & \tagged{twocol}{\displaybreak[3]\\}
			\leq {}& \delta_j \sum_{\ell=j}^{i} \frac{1}{2^{\ell-j}} & \tagged{twocol}{\displaybreak[3]\\}
			\leq 2 {}& \delta_j,
		\end{alignat*}
		where the first inequality follows from the triangle inequality and the second inequality follows from the convergence of the geometric series $\sum_{\ell=0}^{\infty} 2^{-\ell} = 2$.
		Since $\delta_j = \frac{\delta_1}{2^{j-1}}$, $ \varepsilon_j = \frac{\varepsilon_1}{2^{j-1}}$ and $\delta_1=4\varepsilon_1/d\log(4/3)$, it holds that $\delta_j=\frac{4 \varepsilon_j}{d \log (4/3)}$. This completes the proof of the claim.
	\end{proof}
	
Now by the fact that $G_j\in \Pi_n$ is $(\varepsilon_j,s_j)$-hyperfinite, Claim~\ref{claim:freq_diff} and \cref{thm:local_encondig_hypf}, it follows that %
$G_{i+1}$ is $(4 \varepsilon_j \log \frac{4d}{\varepsilon_j}, s_j)$-hyperfinite, for any $j\leq i+1$. %

	In particular, let $i'$ denote the index such that our algorithm outputs $G_{i'}$, i.e., $H^{(n)}=G_{i'}$. For any $\varepsilon>0$, $j_\varepsilon = \min \{ i \mid 1\leq i\leq i', 4 \varepsilon_i \log \frac{4d}{\varepsilon_i} \leq \varepsilon \}$. It is important to note that even though $i'$ might depend on $n$, the index $j_\varepsilon$ is always independent of $n$, and depends only on $\varepsilon$. %
	
	Then we define 
	$$\rho(\varepsilon):=\max\{n_3(d), s_{j_\varepsilon}\}.$$ 
	
	By the above analysis, for any $n\in X$ with $n\geq n_3(d)$, we find an $n$-vertex graph $H^{(n)}\in \Pi_n$ satisfying the following: for any $\varepsilon>0$, there exists $j_\varepsilon$ such that by removing $(4 \varepsilon_{j_\varepsilon} \log \frac{4d}{\varepsilon_{j_\varepsilon}}) \cdot dn \leq \varepsilon dn$ edges, one can decompose $H^{(n)}$ into connected components each of which has size at most $s_{j_\varepsilon}\leq \rho(\varepsilon)$. Thus, it holds that $H^{(n)}$ is $(\varepsilon, \rho(\varepsilon))$-hyperfinite for any $\varepsilon>0$. 
	
	This completes the proof of the lemma.

\end{proof}

\subsection{Every Complement of a Non-Trivially Testable Property Contains a Hyperfinite Subproperty}\label{sec:complement}

We now prove the second part of \cref{thm:main_subproperty}, \ie, the complement of every non-trivially testable property contains a hyperfinite subproperty. The formal definition of non-trivially testable property is given as follows.

\begin{Definition}[non-trivially testable]
	\label{def:nontrivially_testable}
	A graph property $\Pi$ is \emph{non-trivially testable} if it is testable and there exists $\varepsilon >0$ such that the set of graphs that is $\varepsilon$-far from $\Pi$ is infinite. 
\end{Definition}

Note that for a property that is \emph{not} non-trivially testable, for any $\varepsilon>0$, we can always accept all graphs of size $n\geq n_4$, where $n_4:=n_4(\varepsilon)$ is a finite number (that might not be computable) such that there are at most $n_4$ graphs that are $\varepsilon$-far from having the property. For graphs of size smaller than $n_4$, one can simply read the whole graph to test if the graph satisfies the property or not.

The second part of Theorem~\ref{thm:main_subproperty} will follow from the following lemma.
\begin{lemma}
	\label{thm:complement_hypf_subproperty}
	The complement of every non-trivially testable $d$-bounded graph property $\Pi$ contains an infinite $(0, c)$-hyperfinite subproperty, where $c$ depends only on $\Pi$.
\end{lemma}

\begin{proof}
	Since $\Pi$ is non-trivially testable, by \cref{def:nontrivially_testable}, there exists $\varepsilon > 0$ and an infinite set $N \subseteq \mathbb{N}$ such that for every $n \in N$, $\npiepsfar$ is non-empty. Let $\varepsilon > 0$ be the largest value such that $\piepsfar$ contains an infinite number of graphs. Let $\delta = \frac{1}{13 t}$, where $t:=\pqueries(\varepsilon,d)$ denotes the query complexity of $\Pi$. Let $k = \pcqueries(\varepsilon,d)=t^{2t}$. Fix an arbitrary $n \in N$ such that $n\geq n_0$, where $n_0=n_0(\delta,d,k)$ is the number given in Lemma~\ref{thm:blowup_locality}. Let $G_n \in \npiepsfar$ be an arbitrary graph in $\npiepsfar$. Let $H^{(n)}$ be the $(\delta, k)$-blow-up graph of $G_n$. Note that $H^{(n)}$ is $(0, k)$-hyperfinite. Now we claim that $H^{(n)} \notin \Pi$. 
	
	Assume on the contrary that $H^{(n)} \in \Pi$. By \cref{thm:blowup_locality}, $\norm{\kdiskfreq{G_n} - \kdiskfreq{H^{(n)}}} \leq 1.1\delta$. Therefore, by Theorem~\ref{thm:canonical_tester}, the canonical tester for $\Pi$ accepts $G_n$ with probability at least $2/3$, which is a contradiction to the fact that $G_n \in \npiepsfar$. The lemma follows by defining the set $\Pi':=\{ H^{(n)}: n\in N \}$ and $c=k=\pqueries(\varepsilon,d)^{2\pqueries(\varepsilon,d)}$.
\end{proof}

\subsection{Proof of \cref{thm:canonical_tester}}
\label{sec:canonical_tester}
In this section, we give the proof sketch of  \cref{thm:canonical_tester}. The first part (i.e., the transformations from the original tester \( \mathcal{T} \) to the canonical tester \( \mathcal{T}_C \)) of the proof follows from the proof of the canonical testers in \cite{GR11:proximity,CPS16:testing}, and we sketch the main ideas for the sake of completeness. The last part (i.e., how the behaviour of tester \( \mathcal{T}_C \) relates to the frequency vector) of the proof differs from previous work and it is tailored to obtain the characterization as stated in the theorem, which in turn will be suitable for our analysis of the structures of constant-query properties. 
\begin{proof}[Proof Sketch of \cref{thm:canonical_tester}]
	Let \( \mathcal{T} \) be a tester for \( \npi \) on \( n \)-vertex graphs with error probability (reduced to) at most \( \frac{1}{24} \). The query complexity of the tester \( \mathcal{T}\) will be $t:=c\cdot \pqueries(\varepsilon,d)$ for some constant $c>1$, where $\pqueries(\varepsilon,d)$ is the query complexity of the tester for $\Pi$ with error probability at most $\frac13$. We will then transform $\mathcal{T}$ to a canonical tester $\mathcal{T}_C$ in the same way as in the proof of Lemma 3.1 in~\cite{CPS16:testing} (see also \cite{GR11:proximity}).
	
	Slightly more precisely, we first convert \( \mathcal{T} \) into a tester \( \mathcal{T}_1 \) that samples random \disks{t}\ of the input graph and answers all of \( \mathcal{T} \)'s queries using the corresponding subgraph \( H \). That is, it samples a set $S$ of \( t\) vertices and then makes its decision on the basis of the $t$-discs rooted at vertices in $S$ by using uniformly random ordering of vertices and emulating the execution of $\mathcal{T}$ accordingly on the permuted graph. 
	
Then, we convert \( \mathcal{T}_1 \) into a tester \( \mathcal{T}_2 \) whose output depends only on the edges and non-edges in the explored subgraph, the ordering of all explored vertices and its own random coins. This can be done by letting  \( \mathcal{T}_2 \) accept the input graph \( G \) with the average probability that \( \mathcal{T}_1 \) accepts \( G \) over all possible labellings of $H$ with corresponding sequences of queries and answers. 

Next, we convert \( \mathcal{T}_2 \) into the final tester \( \mathcal{T}_3 \) whose output is independent of the ordering of all explored vertices. This can be done by letting \( \mathcal{T}_3 \) accept with probability that is equal to the average of all acceptance probabilities of \( \mathcal{T}_2 \) over all possible relabellings of vertices in $H$.

Finally, we convert $\mathcal{T}_3$ into a tester $\mathcal{T}_C$ that returns the output deterministically according to the unlabeled version of the explored subgraph and its roots. This can be done by letting $\mathcal{T}_C$ accepts the input graph if and only if the probability associated with the explored subgraph $H$ is at least $1/2$. %

By similar arguments in the proof of Lemma 3.1 in~\cite{CPS16:testing}, we can show that $\mathcal{T}_C$ is a tester for $\Pi$ that has error probability at most $1/12$. That is, for each $G\in \Pi_n$, $\mathcal{T}_C$ accepts $G$ with probability at least $1-\frac{1}{12}$. For any graph $G\in \npiepsfar$, $\mathcal{T}_C$ rejects $G$ with probability at least $1-\frac{1}{12}$. Furthermore, note that the query complexity of $\mathcal{T}_C$ is at most $t\cdot d^{t+2}$.

	Now if we let $n_1:= 12 d^{2t} t^2$, then for any $n\geq n_1$, it holds that with probability at least $1-\frac{d^{2t} t^2}{n}\geq 1- \frac{1}{12}$, none of the $t$ sampled $t$-discs will intersect. That is, with probability $1-\frac{1}{12}$, the decision of the tester $\mathcal{T}_C$ will only depend on the structure (or the isomorphic types) of the explored $t$ \emph{disjoint} $t$-discs. 
	
Let $\delta_C=\frac{1}{12t}$. %
We now consider the input graph $G$ satisfying that $\min_{G' \in \npi} \norm{\diskfreq{t}{G} - \diskfreq{t}{G'}} \leq \delta_C$. Let $G'\in \npi$ denote a graph for which this minimum is attained. Note that there is a bijection $\Phi: V(G)\rightarrow V(G')$ such that $\fdisk{t}{G}{v} \ncong \fdisk{t}{G'}{\Phi(v)}$ for at most a $\delta_C$-fraction of the vertices $v \in V(G)$. %
	Recall that $S$ denotes the sample set. Note that for any vertex $v$ that is sampled independently and uniformly at random, the probability that $\fdisk{t}{G}{v} \ncong \fdisk{t}{G'}{v}$ is bounded by the total variation distance of $\diskfreq{t}{G}$ and $\diskfreq{t}{G'}$, which is at most $\delta_C / 2$ by our assumption. By the union bound, the probability that there exists some vertex $v\in S$ with $\fdisk{t}{G}{v} \ncong \fdisk{t}{G'}{\Phi(v)}$ is at most $|S| \cdot \delta_C \leq t \cdot \frac{1}{12t} \leq \frac{1}{12}$. Since $\mathcal{T}_C$ rejects $G'$ with probability at most $\frac{1}{12}$ and the probability that there exists some pair of all $t$ sampled $t$-discs intersecting is at most $\frac{1}{12}$, $\mathcal{T}_C$ rejects $G$ with probability at most $\frac{1}{12} + \frac{1}{12} + \frac{1}{12} = \frac{1}{4}$.
	
	The case when $G$ satisfying that \( \min_{G' \in \npiepsfar} \norm{\diskfreq{t}{G} - \diskfreq{t}{G'}} \leq \delta_C \) can be analyzed analogously. In particular, if $G$ satisfies this condition, then $\mathcal{T}_C$ accepts $G$ with probability at most $\frac{1}{12} + \frac{1}{12} + \frac{1}{12} = \frac{1}{4}$.

	Therefore, $\mathcal{T}_C$ accepts (resp. rejects) $G$ with probability at least $1-\frac{1}{4}>\frac23$, if \( \min_{G' \in \npi} \norm{\diskfreq{t}{G} - \diskfreq{t}{G'}} \leq \delta_C \) (resp. if \( \min_{G' \in \npiepsfar} \norm{\diskfreq{t}{G} - \diskfreq{t}{G'}} \leq \delta_C \)). %
	
	This completes the proof of the theorem.
\end{proof}

\section{Do Testable Non-Hyperfinite Properties Contain Infinitely Many Expanders?}

In the light of the previous result, a natural question is whether every testable infinite property that is not hyperfinite
must contain an infinite subproperty that consists only of expander graphs or graphs that are close to an expander graph.
Unfortunately, such a statement is not true as the aforementioned Theorem~\ref{theorem:no-expander} shows. In the following, we present the proof of Theorem~\ref{theorem:no-expander}.

\begin{proof}[Proof of Theorem~\ref{theorem:no-expander}]
We start by defining the graph property. $\Pi$ consists of all graphs $G=(V,E)$ with maximum degree $d$ that have a single 
connected component with $\lceil |V|/2 \rceil$ vertices and the remaining $\lfloor |V|/2 \rfloor$
connected components are isolated vertices. We observe that $\Pi$ is not hyperfinite as the big connected component
may be an expander graph and so it requires to remove $\Omega(n)$ edges to partition it into small connected components.
Furthermore, it requires to insert $\Omega(n)$ edges to make the graph connected, which is a necessary condition for having expansion greater than $0$.
Finally, we show that the property can be tested with query complexity $O(d/\varepsilon^2)$. 

The algorithm consists of two stages. In the first stage, we sample $O(1/\varepsilon^2)$ vertices uniformly
at random and estimate the number of isolated vertices. 
We reject, if this number differs from $\lfloor |V|/2 \rfloor$ by more than $\varepsilon |V|/8$.
In the second stage, we sample another $O(1/\varepsilon)$ vertices and perform, for every sampled vertex $v$, a BFS until we have explored the whole 
connected component of $v$ or we have explored more than $12/\varepsilon$ vertices. We may assume that the graph contains more than, say, 
$100/\varepsilon$ vertices as otherwise, we can simply query the whole graph.
The tester rejects, if it finds a connected component that is not an isolated vertex.

We now prove that the above algorithm (with proper choice of constants) is a property tester. Our analysis (in particular for the second stage)
uses some ideas that were first introduced in an analysis of a connectivity tester in \cite{GR02:testing}.
We first show that the tester accepts every $G \in \Pi$. For some sufficiently large constant in the $O$-notation we obtain by Chernoff bounds that the first stage of the tester approximates with probability at least $9/10$ such that the number of isolated vertices in $G$ with an additive error 
of $\varepsilon |V|/8$. If this approximation succeeds, the first stage of the tester does not reject. Furthermore, the second stage
never rejects a graph $G\in \Pi$. Thus, the tester accepts with probability at least $9/10$.
Next consider a graph that is $\varepsilon$-far from $\Pi$ and begin with the following claim.
\begin{claim}
\label{claim:structure}
Let $G$ be $\varepsilon$-far from $\Pi$. Then either the number of isolated vertices in $G$ differs by more than $\varepsilon |V|/4$ from $\lfloor |V|/2 \rfloor$ or there are more than $\varepsilon |V|/12$ 
connected components of size at most $12/\varepsilon$ that are not isolated vertices. 
\end{claim}
\begin{proof}
Assume that the claim is not true and there is 
a graph $G$ that is $\varepsilon$-far from $\Pi$, the number of isolated vertices in $G$ differs by at most $\varepsilon |V|/4$
from  $\lfloor |V|/2 \rfloor$ and there are at most $\varepsilon |V|/12$ connected components of size at most $12/\varepsilon$ 
that are not isolated vertices. We will argue that in this case, we can modify at most $\varepsilon dn$ edges to turn $G$ into
a graph that has $\Pi$, which is a contradiction. 
We start with the connected components that are not isolated vertices. We can add a single edge to connect two such components. 
However, we must make sure that we are not violating the degree bound. If both connected components have a vertex of degree at
most $d-1$, we can simply add an edge to connect them. If all vertices of a connected component have degree $d > 1$ then the component
contains a cycle. We can remove an edge from the cycle without destroying connectivity. Thus, we need to modify at most $3$ edges to
connect two connected components. We observe that there are at most $\varepsilon n/12$ connected components of size more than $12/\varepsilon$
and so there are at most $\varepsilon n/6$ connected components that are not isolated vertices. We can create a single connected
component out of them by modifying $\varepsilon n/2$ edges. 
Our previous modifications did not change the number of isolated vertices in $G$, so it still differs by at most $\varepsilon |V|/12$
from $\lfloor |V|/2 \rfloor$. If there are too many isolated vertices, we can connect each of them to the big connected component
with at most $2$ edge modifications resulting in at most $\varepsilon n/2$ modifications. If there are too few isolated vertices, 
we need to disconnect vertices from the big connected
component. For this purpose consider a spanning tree $T$ of the connected component. We will remove a leave of $T$. This can be done 
with $d$ edge modifications and does not change connectivity. Thus we can create exactly  $\lfloor |V|/2 \rfloor$ isolated vertices using 
at most $\varepsilon d n/4$ modifications. 
Overall, the number of modifications is at most $\varepsilon d n$, which proves that the graph was not $\varepsilon$-far from $\Pi$. 
A contradiction.
\end{proof}
It remains to show that our tester rejects any $G$ that is $\varepsilon$-far from $\Pi$. By Claim \ref{claim:structure} we know 
that either the number of isolated vertices in $G$ differs by more than $\varepsilon |V|/4$ from $\lfloor |V|/2 \rfloor$
or $G$ has at least $\varepsilon |V|/12$ connected components of size at most $12/\varepsilon$. 
In the first case, our algorithm rejects with probability at least $9/10$ as it approximates the number of isolated
vertices with additive error $\varepsilon |V|/8$ and rejects if the estimate differs by more than $\varepsilon |V|/4$ 
from  $\lfloor |V|/2 \rfloor$. In the second case we observe that for sufficiently large constant in the $O$-notation
with probability at least $9/10$ we sample a connected component of size at most $12/\varepsilon$. In this case our algorithm
detects the component and rejects. Thus, with probability at least $9/10$ the algorithm rejects.
The query complexity and running time of the algorithm are dominated by the second stage, which can be done in $O(d/\varepsilon^2)$ time.
\end{proof}

Since an expander graph is connected, it follows also that this property contains no graphs that are close to expander graphs.
Consider the \kdisks\ of graphs from the property $\Pi$ in the proof of Theorem \ref{theorem:no-expander}. Recall
that the graphs from the property consist of a connected graph on $\lceil |V|/2 \rceil$ vertices and $\lfloor |V|/2 \rfloor$
isolated vertices. We may view graphs in $\Pi$ as the union of two graphs $G_1$ and $G_2$ of roughly the same size that satisfy 
two different properties: $G_1$ is connected and the $G_2$ has no edges. The \kdisks\ of these graphs have two interesting properties:
\begin{itemize}
\item
no \kdisk\ in $G_1$ occurs in $G_2$ and vice versa, and
\item
their centers cannot be adjacent in any graph. 
\end{itemize}
If $G_1$ and $G_2$ have the above properties then this means that the \kdisks\ cannot \enquote{mix} in any connected component of another graph.
Thus, we know whether they are supposed to come from $G_1$ or $G_2$, which is helpful to design a property tester.
We remark that this phenomenon can also happen for other \kdisks\ like, for example, if $G_1$ is $4$-regular and $G_2$ is $6$-regular.
We believe that understanding this phenomenon is important for a characterization of testable properties in bounded-degree graphs as
we can use it to construct other testable properties in a similar way as above. 
This motivates the following definition:
\begin{Definition}\label{def:incompatible}
We call two \kdisk\ isomorphism types $D_1,D_2$ with roots $u_1,u_2$ \emph{incompatible}, if there exists no graph in which two adjacent 
vertices $u_1$ and $u_2$ have \kdisk\ type $D_1$ and $D_2$, respectively.
\end{Definition}

\section{Partitioning Theorem for Bounded-Degree Graphs}

The fact that there are testable properties that are composed of other properties with disjoint sets of \emph{incompatible} \kdisks\ (see Definition~\ref{def:incompatible}) leads
to the question if we can always decompose the vertex set of a graph into sets such that the \kdisk\ types behave \enquote{similarly}
within each set. A simple partition would be to divide the vertex set according to its \kdisk\ isomorphism type.
But such a partition is meaningless. In the light of previous work, we decided to consider the case that a partition has to have only
a small fraction of the edges between the partition classes. We would like to obtain a partition into sets $S_1,\dots, S_r$ and a set
$T$ (which is a separator), such that no edges are between $S_i$ and $S_j$ for any $i \neq j$  and $T$ is of small size. 
The next question is to specify what it means to behave \enquote{similarly}. One such specification is to ask that the \kdisk\ distribution 
inside the partition is stable for every subset. Obviously, this cannot always be the case unless there is only one \kdisk\ isomorphism
type. Instead, we are only looking at sets that do not have too many outgoing edges. For these subsets we can show that they
always have roughly the same \kdisk\ distribution as their partition.
The formal theorem we prove is the following.

\begin{theorem}\label{thm:partitionoing}
	Let $G=(V,E)$ be a $d$-bounded graph. For every $k\ge 0$ and every $1\ge \delta >0$ the vertex set $V$ can be partitioned 
	into $r\leq f(\delta,d,k)$ subsets $S_1,\cdots, S_r$ and a set $T$ such that 
	\begin{itemize}
		\item
			for every $i \neq j$ there are no edges between $S_i$ and $S_j$,
		\item
			$|T| \le \delta d|V|$,
		\item
			and for every $i$ and every subset $X$ of $S_i$ with $\phi_G(X) \le \delta^2$ it holds that
			\[\norm{\kdiskfreqsub{G}{X} - \kdiskfreqsub{G}{S_i}} \leq 3 \delta. \]
	\end{itemize}
\end{theorem}

\begin{proof}
We will first construct a partition of $V$ into sets $A_1,\dots, A_t$ for some (possibly very large) value of $t$ and a 
set $T$ such that $|T| \le \delta d |V|$ and such that there are no edges between any pair of $A_i$ and $A_j$.
Then we construct each set $S_i$ as a union of some of the sets $A_j$. Finally, we prove that the $S_i$ satisfy the 
third property (the first two follow from the construction of the $A_j$).

We start with $T=\emptyset$ and $W=V$.
Let $A$ be a subset of vertices of $W$ with $\phi_G(A) \le \delta$. We may assume that $A$ contains no proper subset
with this property (otherwise, we take this subset). We put the neighbors of vertices from $A$ that are not in $A$
into the set $T$ and remove $T$ from $W$. We store the set $A$ as $A_1$ and remove it from $W$. We then repeat this process
as long as possible computing the sets $A_2,A_3,\dots$. We observe that every vertex is removed at most once from $W$. 
Whenever we remove a set $A_i$ we move at most $\delta d |A_i|$ neighbors into $T$ since $\phi_G(A_i) \le \delta$.
Hence, $|T| \le \delta d |V|$. Furthermore, we observe that by construction there are no edges between $A_i$ and $A_j$ 
for any $i \not= j$.

It remains to construct the sets $S_i$. For this purpose, we put a $\delta$-net over the space of all \kdisk\ frequency vectors, i.e. 
we compute a smallest set $N=\{v_1,\dots,v_{|N|}\}$ of frequency vectors such that every frequency vector there exists a vector in $N$ within $l_1$ distance at most
$\delta$. We observe that $|N|$ is a function of $k,d$ and $\delta$. %
We then define $S_i$ to be the union of all $A_j$
that have $v_i$ as the closest vector to their frequency vector. It remains to prove that the $S_i$ satisfy the
third property for $\delta^2$. For this purpose consider an arbitrary subset $X \subseteq S_i$.
We consider $X\cap A_j$ for the sets $A_j$ whose union $S_i$ is. If $X\cap A_j \not= A_j$ then we know that $\phi_G(X\cap A_j) >\delta$. Recall that the edges that leave $X \cap A_j$ either go to $A_j \setminus X$ or to $T$, where $X \cap T = \emptyset$.
If $\phi_G(X) \le \delta^2$, then it holds that at most a $\delta$-fraction of the elements from $X$ can be from a 
subset $A_j$ with $\phi_G(X \cap A_j) > \delta$. This is true as otherwise the number of edges crossing $X$ and $V\setminus X$ is at least $\delta|X|\cdot \delta d$, which contradicts the assumption that $\phi_G(X) \le \delta^2$. %
Let $J$ be the set of all indices $j$ such that $A_j \cap X = A_j$. Hence we get
\begin{align*}
	& \kdiskfreqsub{G}{X} & \tagged{twocol}{\displaybreak[3]\\}
	= {}& \sum_{j} \sum_{x\in X\cap A_j} \frac{\kdiskfreqsub{G}{x}}{\lvert X \rvert} & \tagged{twocol}{\displaybreak[3]\\}\untagged{twocol}{\\}
	\untagged{twocol}{&&}= {}& \frac{1}{|X|} \Big( \sum_{j \in J} \sum_{x\in X\cap A_j} \kdiskfreqsub{G}{x} \tagged{twocol}{& \\ & \quad} + \sum_{j \notin J} \sum_{x\in X\cap A_j} \kdiskfreqsub{G}{x} \Big). &
\end{align*}
Now let us define $X_1 = \{ x\in X | x\in A_j, j\in J\}$ and $X_2 = X \setminus X_1$. We know that $|X_2| \leq \delta |X|$. %
We also observe that
\begin{equation*}
	\biggnorm{\frac{1}{|X_1|} \sum_{x\in X_1} \kdiskfreqsub{G}{x} - \kdiskfreqsub{G}{S_i}} \le \delta
\end{equation*}
and
\begin{equation*}
	\biggnorm{\frac{1}{|X_2|} \sum_{x\in X_2} \kdiskfreqsub{G}{x} - \kdiskfreqsub{G}{S_i}} \le 2
\end{equation*}
since all frequency vectors have $l_1$-norm $1$. It follows that %
\begin{alignat*}{10}
	& \bignorm{\kdiskfreqsub{G}{X} - \kdiskfreqsub{G}{S_i}} & \\
	= {}& \biggnorm{\frac{1}{|X|} \cdot \Big( \sum_{x\in X_1} \kdiskfreqsub{G}{x}
	\tagged{twocol}{ & \\ & \enskip}
	+ \sum_{x\in X_2} \kdiskfreqsub{G}{x} \Big) - \kdiskfreqsub{G}{S_i}} & \displaybreak[3]\\
	= {}& \biggnorm{\frac{1}{|X|} \cdot \Big( \sum_{x\in X_1} \kdiskfreqsub{G}{x} - \lvert X_1 \rvert \cdot \kdiskfreqsub{G}{S_i} & \\
	& \enskip + \sum_{x\in X_2}  \kdiskfreqsub{G}{x} - \lvert X_2 \rvert \cdot \kdiskfreqsub{G}{S_i} \Big) } & \displaybreak[3]\\
	\le {}& \frac{|X_1|}{|X|} \cdot \delta + \frac{|X_2|}{|X|} \cdot 2 & \\
	\le {}& 3 \delta.
\end{alignat*}
This completes the proof of the theorem.
\end{proof}

\section{Conclusions}

We have shown that every constant-time testable property in the bounded-degree graph model is either finite or contains an
infinite hyperfinite subproperty. We hope that this result is a first step to obtain a full characterization of all testable
properties in bounded-degree graphs. Unfortunately, a similar result cannot be derived for expander graphs, i.e. it is not
true that every testable infinite property that is not hyperfinite contains an infinite family of expander graphs or graphs
that are close to expander graphs. The structure of this counter-example motivated us to study partitionings of bounded-degree
graphs into sets of vertices such that the distribution of \kdisks\ on any subset with bounded expansion is close to the distribution 
of the set. We hope that this partitioning will be helpful to make further progress towards a characterization of all testable
properties in bounded-degree graphs.

\bibliographystyle{alpha}
\bibliography{BeyondHyperfinite}
\end{document}